\documentclass[11pt]{article}
\usepackage[utf8]{inputenc}

\usepackage[backend=biber]{biblatex}
\usepackage{
  fullpage, microtype, authblk,
  booktabs, diagbox, algorithm,
  amsfonts, amsmath, amssymb, amsthm, bm,
  url, tikz, tikz-qtree, subcaption, relsize, xspace,
  graphicx,
}

\usepackage[font=small,labelfont=bf]{caption}
\usepackage[inline]{enumitem}
\usepackage[normalem]{ulem}
\usepackage[noend]{algpseudocode}
\usepackage{nicefrac}
\usepackage{xcolor}
\usepackage[title]{appendix}
\usepackage[none]{hyphenat}

\usetikzlibrary{arrows.meta, calc, positioning, matrix}
\MakeRobust{\Call}

\newtheorem{theorem}{Theorem}

\newtheorem{lemma}[theorem]{Lemma}

\newcommand{\mund}{\mathunderscore}

\newcommand{\G}{\ensuremath{\mathcal{G}}\xspace}
\newcommand{\I}{\ensuremath{\mathcal{I}}\xspace}
\newcommand{\T}{\ensuremath{\mathcal{T}}\xspace}
\newcommand{\J}{\ensuremath{\mathcal{J}}\xspace}

\newcommand{\alga}[3]{\texttt{add\mund{}contact($#1$,\;$#2$,\;$#3$)}\xspace}
\newcommand{\algb}[4]{\texttt{can\mund{}reach($#1$,\;$#2$,\;$#3$,\;$#4$)}\xspace}
\newcommand{\algc}[2]{\texttt{is\mund{}connected($#1$,\;$#2$)}\xspace}
\newcommand{\algd}[4]{\texttt{reconstruct\mund{}journey($#1$,\;$#2$,\;$#3$,\;$#4$)}\xspace}

\newcommand{\pinsert}[3]{\textsc{insert($#1$,\;$#2$,\;$#3$)}\xspace}
\newcommand{\pprev}[2]{\textsc{find\mund{}prev($#1$,\;$#2$)}\xspace}
\newcommand{\pnext}[2]{\textsc{find\mund{}next($#1$,\;$#2$)}\xspace}
\newcommand{\pinsertname}{\textsc{insert}\xspace}
\newcommand{\pprevname}{\textsc{find\mund{}prev}\xspace}
\newcommand{\pnextname}{\textsc{find\mund{}next}\xspace}

\newcommand{\baccess}[2]{\textsc{access($#1$, $#2$)}\xspace}
\newcommand{\brank}[3]{\textsc{rank$_#1$($#2$, $#3$)}\xspace}
\newcommand{\bselect}[3]{\textsc{select$_#1$($#2$, $#3$)}\xspace}
\newcommand{\binsert}[3]{\textsc{insert$_#1$($#2$, $#3$)}\xspace}
\newcommand{\bupdate}[3]{\textsc{update$_#1$($#2$, $#3$)}\xspace}
\newcommand{\bremove}[2]{\textsc{remove($#1$, $#2$)}\xspace}
\newcommand{\binsertw}[3]{\textsc{insert\mund{}word$_{#1}$($#2$, $#3$)}\xspace}
\newcommand{\bremovew}[3]{\textsc{remove\mund{}word$_{#1}$($#2$, $#3$)}\xspace}
\newcommand{\bunsetrange}[1]{\textsc{unset\mund{}one\mund{}range($#1$, $j_1$, $j_2$)}\xspace}
\newcommand{\bjoin}[2]{\textsc{join($#1$, $#2$)}\xspace}
\newcommand{\bsplit}[2]{\textsc{split\mund{}at\mund{}jth\mund{}one($#1$, $#2$)}\xspace}
\newcommand{\bjoinname}{\textsc{join}\xspace}
\newcommand{\bsplitname}{\textsc{split\mund{}at\mund{}jth\mund{}one}\xspace}

\newcommand{\Btree}{B$^+$-tree\xspace}
\newcommand{\Btrees}{B$^+$-trees\xspace}
\newcommand{\bv}{bit-vector\xspace}
\newcommand{\bvs}{bit-vectors\xspace}

\tikzset{
vertex/.style={
		draw,
		fill=white,
		circle,
		semithick,
		inner sep=0pt,
		minimum size=7pt,
		font=\footnotesize,
	},
arc/.style={
-{Latex[length=5pt]},
semithick,
font=\scriptsize,
},
}

\newcommand{\drawintervalsubfig}[3]{
	\tikzset{
		cell/.style={
				draw,
				font=\scriptsize,
				minimum width=0.5cm,
				minimum height=0.45cm
			}
	}

	\pgfmathsetmacro{\lastdep}{0}
	\pgfmathsetmacro{\lastarr}{0}
	\foreach \dep\arr\lab[
		evaluate=\dep as \xdep using \dep*0.5,
		evaluate=\arr as \xarr using \arr*0.5
	] in #1 {
			\pgfmathtruncatemacro\dstart{\lastdep+1}
			\pgfmathtruncatemacro\dend{\dep-1}
			\ifnum \dstart<\dep
				\foreach \i [evaluate=\i as \xi using \i*0.5] in {\dstart, ..., \dend} {
						\node (n\i1) [cell] at (\xi, -0.7) {\texttt{0}};
					}
			\fi
			\node (n\dep1) [cell] at (\xdep, -0.7) {\texttt{1}};

			\pgfmathtruncatemacro\astart{\lastarr+1}
			\pgfmathtruncatemacro\aend{\arr-1}
			\ifnum \astart<\arr
				\foreach \i [evaluate=\i as \xi using \i*0.5] in {\astart, ..., \aend} {
						\node (n\i2) [cell] at (\xi, -1.4) {\texttt{0}};
					}
			\fi
			\node (n\arr2) [cell] at (\xarr, -1.4) {\texttt{1}};
			\global\let\lastdep=\dep
			\global\let\lastarr=\arr
		}

	\ifnum \lastdep<6
		\pgfmathtruncatemacro\dstart{\lastdep+1}
		\foreach \i [evaluate=\i as \xi using \i*0.5] in {\dstart, ..., 6} {
				\node (n\i1) [cell] at (\xi, -0.7) {\texttt{0}};
			}
	\fi

	\ifnum \lastarr<6
		\pgfmathtruncatemacro\astart{\lastarr+1}
		\foreach \i [evaluate=\i as \xi using \i*0.5] in {\astart, ..., 6} {
				\node (n\i2) [cell] at (\xi, -1.4) {\texttt{0}};
			}
	\fi

	\foreach \pos [evaluate=\pos as \x using \pos*0.5] in {1, ..., 6} {
			\node at (\x, -1.9) {\color{black!80}\tiny\pos};
		}

	\pgfmathsetmacro{\ys}{0.2}
	\pgfmathsetmacro{\mcolor}{0}
	\foreach \dep\arr\lab[count=\i] in #1 {
			\ifnum \i=1
				\ifnum #2=1
					\pgfmathsetmacro\mcolor{100}
				\fi
			\fi

			\draw[red!\mcolor!black,semithick] ($(n\dep1.north) + (0, \ys)$) -- ($(n\arr1.north) + (0, \ys)$) node[below right = -0.24 and 0.001] {\tiny$\mathcal{I}_\lab$};
			\pgfmathsetmacro\nextys{\ys+0.2}
			\global\let\ys=\nextys
		}

	\ifnum #3=1
		\node[left = 0.2 of n11] {\small $D$};
		\node[left = 0.2 of n12] {\small $A$};
	\fi
}

\title{Dynamic Compact Data Structure for Temporal Reachability with Unsorted Contact Insertions}

\author[1]{Luiz Fernando Afra Brito}
\author[1]{Marcelo Keese Albertini}
\author[1]{Bruno Augusto Nassif Travençolo}
\author[2]{Gonzalo Navarro}

\affil[1]{Faculty of Computer Science (FACOM), Federal University of Uberlândia, Uberlândia, Brazil}
\affil[2]{IMFD \& Department of Computer Science (DCC), University of Chile, Santiago, Chile}

\begin{document}

\maketitle

\begin{abstract}
	Temporal graphs represent interactions between entities over time.
	Deciding whether entities can reach each other through temporal paths is useful for various applications such as in communication networks and epidemiology.
	Previous works have studied the scenario in which addition of new interactions can happen at any point in time.
	A known strategy maintains, incrementally, a Timed Transitive Closure by using a dynamic data structure composed of $O(n^2)$ binary search trees containing non-nested time intervals.
	However, space usage for storing these trees grows rapidly as more interactions are inserted.
	In this paper, we present a compact data structures that represent each tree as two dynamic bit-vectors.
	% Moreover, we propose algorithms to insert new intervals while preserving the non-containment property.
	In our experiments, we observed that our data structure improves space usage while having similar time performance for incremental updates when comparing with the previous strategy in temporally dense temporal graphs.
\end{abstract}

\section{Introduction}

Temporal graphs represent interactions between entities over time.
These interactions often appear in the form of contacts at specific timestamps.
Moreover, entities can also interact indirectly with each other by chaining several contacts over time.
For example, in a communication network, devices that are physically connected can send new messages or propagate received ones; thus, by first sending a new message and, then, repeatedly propagating messages over time, remote entities can communicate indirectly.
Time-respecting paths in temporal graphs are known as temporal paths, or simply journeys, and when a journey exists from one entity to another, we say that the first can reach the second.

In a computational environment, it is often useful to check whether entities can reach each other while using low space.
Investigations on temporal reachability have been used, for instance,
for characterizing mobile and social networks~\cite{tang2010characterising,Linhares2019a},
and for validating protocols and better understanding communication networks~\cite{cacciari1996atemporal, whitbeck2012temporal}.
Some other applications require the ability to reconstruct a concrete journey if one exists.
Journey reconstruction has been used in applications such as
finding and visualizing detailed trajectories in transportation networks~\cite{wu2017mining, betsy2007spatio, zeng2014visualizing},
and matching temporal patterns in temporal graph databases~\cite{vera2016querying,LVM18}.
In all these applications, low space usage is important because it allows the maintenance of larger temporal graphs in primary memory.
% Furthermore, reachability queries and journey reconstruction are primitive operations for many other higher-level queries in temporal graphs~\cite{}.
% todo Need to find some refs and put in the text

In~\cite{barjon2014testing,whitbeck2012temporal}, the authors considered updating reachability information given a chronologically sorted sequence of contacts.
In this problem, a standard Transitive Closure (TC) is maintained as new contacts arrive.
Differently, in~\cite{brito2022dynamic, wu2016reachability}, the authors studied the problem in which sequences of contacts may be chronologically unsorted and queries may be intermixed with update operations.
% Data structures for the later problem are very important.
For instance, during scenarios of epidemics, outdated information containing interaction details among infected and non-infected individuals are reported in arbitrary order, and the dissemination process is continually queried in order to take appropriate measures against contamination~\cite{Ponciano2021,xiao2018reconstructing,enright2021deleting,rozenshtein2016reconstructing}.

Particularly to our interest, the data structure proposed by~\cite{brito2022dynamic} maintains a Timed Transitive Closure (TTC), a generalization of a TC that takes time into consideration.
It maintains well-chosen sets of time intervals describing departure and arrival timestamps of journeys in order to provide time related queries and enable incremental updates on the data structure.
The key idea is that, each set associated with a pair of vertices only contains non-nested time intervals and it is sufficient to implement all the TTC operations.
Our previous data structure maintains only $O(n^2\tau)$ intervals (as opposed to $O(n^2\tau^2)$) using $O(n^2)$ dynamic Binary Search Trees (BSTs).
Although the reduction of intervals is interesting, the space to maintain $O(n^2)$ BSTs containing $O(\tau)$ intervals each can still be prohibitive for large temporal graphs.

In this paper, we propose a dynamic compact data structure to represent TTCs incrementally while answering reachability queries.
Our new data structure maintains each set of non-nested time intervals as two dynamic \bvs, one for departure and the other for arrival timestamps.
Each dynamic \bv{} uses the same data layout introduced in~\cite{prezza2017framework}, which resembles a \Btree~\cite{bplustree} with static \bvs{} as leaf nodes.
In this work, we used a raw \bv{} representation on leaves that stores bits as a sequence of integer words.
In our experiments, we show that our new algorithms follow the same time complexities introduced in the previous section, however, the space to maintain our data structure is much smaller on temporally dense temporal graphs.
Encoding~\cite{elias1975universal} or packing~\cite{lemire2015decoding} the distance between 1's on leaves may improve the efficiency on temporally very sparse temporal graphs.

\subsection{Organization of the document}
This paper is organized as follows.
In Section~\ref{sec:background}, we briefly review the Timed Transitive Closure, the data structure introduced in~\cite{brito2022dynamic}, and the dynamic \bv proposed by~\cite{prezza2017framework}.
% In Section~\ref{sec:timed-transitive-closure}, we introduce timed transitive closures, study their basic properties, and provide a number of low-level primitives for manipulating them.
In Section~\ref{sec:compact-data-structure}, we describe our data structure along with the algorithms for each operation.
In Section~\ref{sec:experiments}, we conduct some experiments comparing our data structure with the previous work~\cite{brito2022dynamic}.
Finally, Section~\ref{sec:conclusions} concludes with some remarks and open questions such as the usage of an encoding or packing techniques for temporal very sparse temporal graphs.

\section{Background}\label{sec:background}

\subsection{Timed Transitive Closure}

Following the formalism in~\cite{casteigts2012time}, a temporal graph is represented by a tuple $\G = (V, E, \T, \rho, \zeta)$ where:
$V$ is a set of vertices;
$E \subseteq V \times V$ is a set of edges;
$\T$ is the time interval over which the temporal graph exists (lifetime);
$\rho: E \times \T \to\{0, 1\}$ is a function that expresses whether a given edge is present at a given time instant;
and $\zeta: E \times\T \mapsto \mathbb{T}$ is function that expresses the duration of an interaction for a given edge at a given time, where $\mathbb{T}$ is the time domain.
In this paper, we consider a setting where $E$ is a set of directed edges, $\mathbb{T}$ is discrete such that $\T = [1, \tau] \subseteq \mathbb{T}$ is the lifetime containing $\tau$ timestamps, and $\zeta = \delta$, where $\delta$ is any fixed positive integer.
Additionally, we call $(u, v, t)$ a contact in \G if $\rho((u, v), t) = 1$.

Reachability in temporal graphs can be defined in terms of journeys.
A journey from $u$ to $v$ in \G is a sequence of contacts $\J = \langle c_1, c_2, \ldots, c_k \rangle$, whose sequence of underlying edges form a valid $(u,v)$-path in the underlying graph $G$ and, for each contact $c_i = (u_i, v_i, t_i)$, it holds that $t_{i+1} \ge t_i + \delta$ for $i \in [1, k-1]$.
Throughout this article we use $departure (\J) = t_1$, and $arrival (\J) = t_{k} + \delta$.
Thus, a vertex $u$ can reach a vertex $v$ within time interval $[t_1, t_2]$ iff there exists a journey $\J$ from $u$ to $v$ such that $ t_1 \leq departure(\J) \leq arrival(\J) \leq t_2$.

\begin{figure}
	\centering
	\begin{tikzpicture}[scale=2]
		\node[vertex] at (162:8mm) (a) {}
		node[left = -2pt of a] {$a$};
		\node[vertex] at (90:7mm)  (b) {}
		node[above= -2pt of b] {$b$};
		\node[vertex] at (270:2mm) (c) {}
		node[below= -2pt of c] {$c$};
		\node[vertex] at (378:8mm) (d) {}
		node[right= -2pt of d] {$d$};
		\tikzstyle{every node}=[inner sep=1pt, font=\small]
		\draw (a) edge[arc, bend left=15] node[above left] {$1$} (b);
		\draw (a) edge[arc, bend right=15] node[below right] {$2$} (b);
		\draw (b) edge[arc] node[above right] {$3$} (d);
		\draw (c) edge[arc] node[below left] {$4$} (a);
		\draw (c) edge[arc] node[below right] {$5$} (d);
	\end{tikzpicture}~~~~
	\begin{tikzpicture}[scale=2]
		\node[vertex] at (162:8mm) (a) {}
		node[left = -2pt of a] {$a$};
		\node[vertex] at (90:7mm)  (b) {}
		node[above= -2pt of b] {$b$};
		\node[vertex] at (270:2mm) (c) {}
		node[below= -2pt of c] {$c$};
		\node[vertex] at (378:8mm) (d) {}
		node[right= -2pt of d] {$d$};
		\draw (a) edge[arc, bend left=15] node[above left] {$[1, 2]$} (b);
		\draw (a) edge[arc, bend right=15] node[below right= -5pt and 0] {$[2, 3]$} (b);
		\draw (a) edge[arc] node[below] {$[2, 4]$} (d);
		\draw (b) edge[arc] node[above right] {$[3, 4]$} (d);
		\draw (c) edge[arc] node[below left] {$[4, 5]$} (a);
		\draw (c) edge[arc] node[below right] {$[5, 6]$} (d);
	\end{tikzpicture}
	\caption{On the left, a temporal graph $\G$ on four vertices $V = \{ a, b, c, d \}$, where the presence times of edges are depicted by labels. For $\delta = 1$, this temporal graph has only two non-trivial journeys, \textit{i.e.} journeys with more than one contact, namely $\mathcal{J}_1 = \langle(a, b, 1), (b, d, 4)\rangle$ and $\mathcal{J}_2 = \langle(a, b, 2), (b, d, 4)\rangle$. On the right, the corresponding Timed Transitive Closure (TTC). Note that only the interval $\I_2 = [2, 4]$, regarding $\J_2$, is depicted on the edge from $a$ to $d$ because the other possibility, $\I_1 = [1, 4]$, regarding $\J_1$, encloses $I_2$. A query to check whether $a$ reaches $d$ within the time interval $\I_1$ can also be satisfied by using $\I_2$.}
	\label{fig:timed-transitive-closure}
\end{figure}
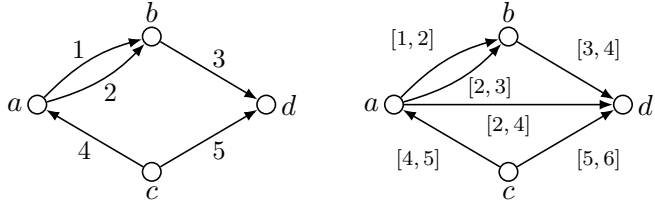

In~\cite{brito2022dynamic}, the authors introduced the Timed Transitive Closure (TTC), a transitive closure that captures the reachability information of a temporal graph within all possible time intervals.
Informally, the TTC of a temporal graph \G is a multigraph with time interval labels on edges.
Each time interval expresses the $departure(\J)$ and $arrival(\J)$ timestamps of a journey \J in \G as its left and right boundaries, respectively.
This additional information allows answering reachability queries parametrized by time intervals and also deciding if a new contact occurring anywhere in history can be composed with existing journeys.
Furthermore, a TTC needs at most $\tau$ edges (in the same direction) between two vertices instead of $\tau^2$ to perform basic operations.
The key idea is that each set of intervals from these edge labels can be reduced to a set containing only non-nested time intervals.
For instance, in the contrived example shown in Figure~\ref{fig:timed-transitive-closure}, we can see that, even though the information of an existing journey in the temporal graph was discarded in the corresponding TTC, a reachability query that could be satisfied by a ``larger'' interval can also be satisfied by a ``smaller'' nested interval.

Their data structure encodes TTCs as $n \times n$ matrices, in which every entry $(u,v)$ points to a self-balanced Binary Search Tree (BST) denoted by $T(u,v)$.
Each tree $T(u, v)$ contain up to $\tau$ intervals corresponding to the reduced edge labels from vertex $u$ to vertex $v$ in the TTC.
As all these intervals are non-nested, one can use any of their boundaries (departure or arrival) as sorting key.
This data structure supports the following operations:
\alga{u}{v}{t}, which updates information based on a new contact $(u, v, t)$;
\algb{u}{v}{t_1}{t_2}, which returns true if $u$ can reach $v$ within the interval $[t_1, t_2]$;
\algc{t_1}{t_2}, which returns true if $\mathcal{G}$, restricted to the interval $[t_1,t_2]$, is temporally connected, \textit{i.e.}, all vertices can reach each other within $[t_1, t_2]$; and
\algd{u}{v}{t_1}{t_2}, which returns a journey (if one exists) from $u$ to $v$ occurring within the interval $[t_1, t_2]$.
All these operations can be implemented using the following BST primitives, where $T_{(u, v)}$ is a BST containing reachability information regarding journeys from $u$ to $v$:
\pprev{T_{(u, v)}}{t}, which retrieves from $T_{(u, v)}$ the earliest interval $[t^-, t^+]$ such that $t^- \geq t$, if any, and nil otherwise; 
\pnext{T_{(u, v)}}{t}, which retrieves from $T_{(u, v)}$ the latest interval $[t^-, t^+]$ such that $t^+ \leq t$, if any, and nil otherwise; and
\pinsert{T_{(u, v)}}{t^-}{t^+}, which inserts into $T_{(u, v)}$ a new interval $\I = [t_1, t_2]$ if no other interval $\I'$ such that $\I \subseteq \I'$ exists while removing all intervals $\I''$ such that $\I \subseteq \I'' $.
% Following their algorithms, \pinsertname takes $O(\log\tau + d)$ time, where $d$ is the number of redundant intervals removed, and \pprevname and \pnextname takes $O(\log{\tau})$ time.

The algorithm for \alga{u}{v}{t} manages the insertion of a new contact $(u, v, t)$ as follows.
First, the interval $[t, t + \delta]$, corresponding to the trivial journey \J from $u$ to $v$ with $departure(\J) = t$ and $arrival(\J) = t + \delta$, is inserted in $T_{(u,v)}$ using the \pinsertname primitive, which runs in time $O(\log\tau + d)$ where $d$ is the number of redundant intervals removed.
Then, the core of the algorithm consists of computing the indirect consequences of this insertion for the other vertices.
Their algorithm consists of enumerating these compositions with the help of the \pprevname and \pnextname primitives, which runs in time $O(\log\tau)$, and inserting them into the TTC using \pinsertname.
As there can only be one new interval for each pair of vertices, the algorithm takes $O(n^2\log\tau)$ amortized time.

The algorithm for \algb{u}{v}{t_1}{t_2} consists of testing whether $T_{(u, v)}$ contains at least one interval included in $[t_1, t_2]$.
The cost of this algorithm reduces essentially to calling \pnext{T_{(u, v)}}{t_1} once, which takes $O(\log\tau)$ time.
The algorithm for \algc{t_1}{t_2} simply calls \algb{u}{v}{t_1}{t_2} for every pair of vertices; therefore, it takes $O(n^2 \log\tau)$ time.
For \algd{u}{v}{t_1}{t_2}, an additional field successor must be included along every time interval indicating which vertex comes next in (at least one of) the journeys.
The algorithm consists of unfolding intervals and successors, one pair at a time using the \pnextname primitive, until the completion of the resulting journey of length $k$; therefore, it takes $O(k\log\tau)$ time in total.

\subsection{Dynamic \bvs}

A \bv $B$ is a data structure that holds a sequence of bits and provides the following operations:
\baccess{B}{i}, which accesses the bit at position $i$;
\brank{b}{B}{i}, which counts the number of $b$'s until (and including) position $i$; and
\bselect{b}{B}{j}, which finds the position of the $j$-th bit with value $b$.
It is a fundamental data structure to design more complex data structures such as compact sequence of integers, text, trees, and graphs~\cite{navarro2016compact,caro2016compressed}.
Usually, \bvs are static, meaning that we first construct the data structure from an already known sequence of bits in order to take advantage of these query operations.

Additionally, a dynamic \bv allows changes on the underlying bits.
Although many operations to update a dynamic \bv has been proposed, the following are the most commonly used:
\binsert{b}{B}{i}, which inserts a bit $b$ at position $i$;
\bupdate{b}{B}{i}, which writes the new bit $b$ to position $i$; and
\bremove{B}{i}, which removes the bit at position $i$.
Apart from these operations, there are others such as \binsertw{w}{B}{i}, which inserts a word $w$ at position $i$, and \bremovew{n}{B}{i}, which removes a word of $n$ bits from position $i$.

\begin{figure*}
	\centering
	\begin{tikzpicture}[
			inner sep = 0pt,
			marray/.style={matrix of nodes,
					nodes={rectangle,draw,minimum width=1.2em, inner sep = 2pt, font=\footnotesize},
					ampersand replacement=\&},
			mbv/.style={draw,minimum width=1.2em, inner sep = 2pt, font=\footnotesize },
			edge from parent/.style= {draw, edge from parent path={(\tikzparentnode) -- (\tikzchildnode)}} ]
		\tikzset{level 1/.style = {level distance = 40pt, sibling distance = 20pt}}
		\tikzset{level 2/.style = {level distance = 25pt, sibling distance = 10pt}}
		\Tree [.\node(a){\tikz{\matrix[marray] {$16$ \& $12$ \& $12$ \\ $3$ \& $2$ \& $5$ \\}}};
		[.\node(b){\tikz{\matrix[marray] {$4$ \& $4$ \& $4$ \& $4$ \\ $1$ \& $1$ \& $0$ \& $1$ \\}}};
		\node[mbv]{\texttt{1000}};
		\node[mbv]{\texttt{0010}};
		\node[mbv]{\texttt{0000}};
		\node[mbv]{\texttt{0100}}; ]
		[.\node(c){\tikz{\matrix[marray] {$4$ \& $4$ \& $4$ \\ $0$ \& $2$ \& $0$ \\}}};
		\node[mbv]{\texttt{0000}};
		\node[mbv]{\texttt{1010}};
		\node[mbv]{\texttt{0000}}; ]
		[.\node(d){\tikz{\matrix[marray] {$4$ \& $4$ \& $4$ \\ $3$ \& $1$ \& $1$ \\}}};
		\node[mbv]{\texttt{1011}};
		\node[mbv]{\texttt{1000}};
		\node[mbv]{\texttt{0001}}; ] ]

		\node[left= 1.5mm of a, yshift=2mm, font=\footnotesize] {$num = $};
		\node[left= 1.5mm of a, yshift=-2mm, font=\footnotesize] {$ones = $};
		\node[left= 1.5mm of b, yshift=2mm, font=\footnotesize] {$num = $};
		\node[left= 1.5mm of b, yshift=-2mm, font=\footnotesize] {$ones = $};
		\node[left= 1.5mm of c, yshift=2mm, font=\footnotesize] {$num = $};
		\node[left= 1.5mm of c, yshift=-2mm, font=\footnotesize] {$ones = $};
		\node[left= 1.5mm of d, yshift=2mm, font=\footnotesize] {$num = $};
		\node[left= 1.5mm of d, yshift=-2mm, font=\footnotesize] {$ones = $};
	\end{tikzpicture}
	\caption{A dynamic \bv using the data structure introduced in~\cite{prezza2017framework}. Leaves wrap static \bvs and internal nodes contain pointers to children along with the number of 1's and the total number of bits in each of them. The maximum number of pointers in each internal node $m$ and the length of each static \bv $n$ in this example is $4$.}
	\label{fig:dynamicbv}
\end{figure*}
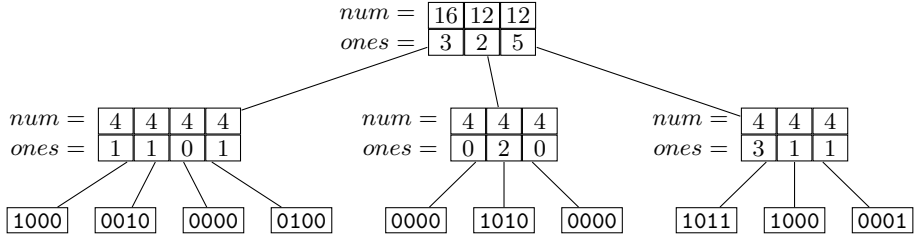

In~\cite{prezza2017framework}, the authors proposed a dynamic data structure for \bvs with a layout similar to \Btrees~\cite{bplustree}.
Leaves wrap static \bvs of maximum length $l$ and internal nodes contain at most $m$ pointers to children along with the number of 1's and the total number of bits in each subtree.
With exception to the root node, static \bvs have a minimum length of $\lceil\nicefrac{l}{2}\rceil$ and internal nodes have at least $\lceil\nicefrac{m}{2}\rceil$ pointers to children.
These parameters serve as rules to balance out tree nodes during insertion and removal of bits.
Figure~\ref{fig:dynamicbv} illustrates the overall layout of this data structure.

Any static \bv representation can be used as leaves, the simplest one being arrays of words representing bits explicitly.
In this case, the maximum length could be set to $l = \Theta(|w|^2)$, where and $|w|$ is the integer word size.
Other possibility is to represent \bvs sparsely by computing the distances between consecutive 1's and then encoding them using an integer compressor such as Elias-Delta~\cite{elias1975universal} or simply packing them using binary packing~\cite{lemire2015decoding}.
In this case, we can instead use as parameter the maximum number of 1's encoded by static \bvs to balance out leaves.

Their data structure supports the main dynamic \bv operations as follows.
An \baccess{B}{i} operation is done by traversing the tree starting at the root node.
In each node the algorithm searches from left to right for the branch that has the $i$-th bit and subtracts from $i$ the number of bits in previous subtrees.
After traversing to the corresponding child node, the new $i$ is local to that subtree and the search continues until reaching the leaf containing the $i$-th bit.
At a leaf node, the algorithm simply accesses and returns the $i$-th local bit in the corresponding static \bv.
If bits in static \bvs are encoded, an additional decoding step is necessary.

The \brank{b}{B}{i} and \bselect{b}{B}{j} operations are similar to \baccess{B}{i}.
For \brank{b}{B}{i}, the algorithm also sums the number of 1's in previous subtrees when traversing the tree.
At a leaf, it finally sums the number of 1's in the corresponding static \bv up to the $i$-th local bit using \texttt{popcount} operations, which counts the number of 1's in a word, and return the resulting value.
For \bselect{b}{B}{j}, the algorithm instead uses the number of 1's in each subtree to guide the search.
Thus, when traversing down, it subtracts the number of 1's in previous subtrees from $j$, and sums the total number of bits.
At a leaf, it searches for the position of the $j$-th local set bit using \texttt{clz} or \texttt{ctz} operations, which counts, respectively, the number of leading and trailing zeros in a word; sums it, and returns the resulting value.

The algorithm for \binsert{b}{B}{i} first locates the leaf that contains the static \bv with the $i$-th bit.
During this top-down traversal, it increments the total number of bits and the number of 1's, whether $b = 1$, in each internal node key associated with the child it descends.
Then, it reconstructs the leaf while including the new bit $b$.
If the leaf becomes full, the algorithm splits its content into two \bvs and updates its parent accordingly while adding a new key and a pointer to the new leaf.
After this step, the parent node can also become full and, in this case, it must also be split into two nodes.
Therefore, the algorithm must traverse back, up to the root node, balancing any node that becomes full.
If the root node becomes full, then it creates a new root containing pointers to the split nodes along with the keys associated with both subtrees.

The algorithm for \bremove{B}{i} also has a top-down traversal to locate and reconstruct the appropriate leaf, and a bottom-up phase to rebalance tree nodes.
However, internal node keys associated with the child it descends must be updated during the bottom-up phase since the $i$-bit is only known after reaching the corresponding leaf.
Moreover, a node can become empty when it has less than half the maximum number of entries.
In this case, first, the algorithm tries to share the content of siblings with the current node while updating parent keys.
If sharing is not possible, it merges a sibling into the current node and updates its parent while removing the key and pointer previously related to the merged node.
If the root node becomes empty, the algorithm removes the old root and makes its single child the new root.

The \bupdate{b}{B}{i} operation can be implemented by calling \bremove{B}{i} then \binsert{b}{B}{i}, or by using a similar strategy with a single traversal.
% We note that the implementation of the static \bvs on leaves can change time and space trade-offs.
% A simple static \bv represents the sequence of bits explicitly as an array of words.
% To access a bit, we first access the word containing that bit and than extract the corresponding bit using a bit mask operation.
% To answer rank at a position $i$, we can count the number of bits using fast-to-compute popcount operations of words antil the one containing the $i$-th bit and the use a shift operation in the last word before using a popcount operation.
% To answer select of the $j$-th 1 bit, we can count the number of bits of words until the number surpasses $j$ and than use a clz operation while shifting bits until finding the $j$-th bit.

\section{Dynamic compact data structure for temporal reachability}\label{sec:compact-data-structure}

Our new data structure uses roughly the same strategy as in the previous work~\cite{brito2022dynamic}.
The main difference is the usage of a compact dynamic data structure to maintain a set of non-nested intervals instead of Binary Search Trees (BSTs).
This compact representation provides all BST primitives in order to incrementally maintain Temporal Transitive Closures (TTCs) and answer reachability queries. 
In~\cite{brito2022dynamic}, the authors defined them as follows, where $T_{(u, v)}$ represents a BST holding a set of non-nested intervals associated with the pair of vertices $(u, v)$.
(1) \pnext{T_{(u, v)}}{t} returns the earliest interval $[t^-, t^+]$ in $T_{(u, v)}$ such that $t^- \geq t$, if any, and nil otherwise;
(2) \pprev{T_{(u, v)}}{t} returns the latest interval $[t^-, t^+]$ in $T_{(u, v)}$ such that $t^+ \leq t$, if any, and nil otherwise; and
(3) \pinsert{T_{(u, v)}}{t^-}{t^+} inserts the interval $[t^-, t^+]$ in $T_{(u, v)}$ and performs some operations for maintaining the property that all intervals in $T_{(u, v)}$ are minimal.
% Following their algorithms, \pinsertname takes $O(\log\tau + d)$ time, where $d$ is the number of redundant intervals removed, and \pprevname and \pnextname takes $O(\log{\tau})$ time.

For our new compact data structure, we take advantage that every set of intervals only contains non-nested intervals, thus we do not need to consider other possible intervals.
For instance, if there is an interval $\I = [4, 6]$ in a set, no other interval starting at timestamp $4$ or ending at $6$ is possible, otherwise, there would be some interval $\I'$ such that $\I' \subseteq \I$ or $\I \subseteq \I'$.
Therefore, we can represent each set of intervals as a pair of dynamic \bvs $D$ and $A$, one for departure and the other for arrival timestamps.
Both \bvs{} must provide the following low-level operations: \baccess{B}{i}, \brank{b}{B}{i}, \bselect{b}{B}{j}, \binsert{b}{B}{i}, and \bupdate{b}{B}{i}.

By using these simple \bvs{} operations, we first introduce algorithms for the primitives \pnext{(D, A)_{(u, v)}}{t}, \pprev{(D, A)_{(u, v)}}{t} and \pinsert{(D, A)_{(u, v)}}{t^-}{t^+} that runs, respectively, in time $O(\log{\tau})$, $O(\log{\tau})$ and $O(d\log{\tau})$, where $d$ is the number of intervals removed during an interval insertion.
Note that, now, these operations receive as first argument a pair containing two \bvs{} $D$ and $A$ associated with the pair of vertices $(u, v)$ instead of a BST $T_{(u, v)}$.
If the context is clear, we will simply use the notation $(D, A)$ instead of $(D, A)_{(u, v)}$.

Then, in order to improve the time complexity of \pinsert{(D, A)_{(u, v)}}{t^-}{t^+} to $O(\log{\tau} + d)$, we propose a new \bv{} operation: \bunsetrange{B}, which clears all bits in the range $[\bselect{1}{B}{j_1}, \bselect{1}{B}{j_2}]$.

% In Section~\ref{ssec:compact-binary-tree} we describe our new compact data structure along with algorithms for \pinsert, \pprev and \pnext.
% In Section~\ref{ssec:new-bv-operations} we introduce an algorithm for the \bunsetrange{B} operation on \bvs in order to improve \pinsert.

\subsection{Compact representation of non-nested intervals}%
\label{ssec:compact-binary-tree}

Each set of non-nested intervals is represented as a pair of dynamic \bvs $D$ and $A$, one storing departure timestamps and the other arrival timestamps.
Given a set of non-nested intervals $\I_1, \I_2, \ldots, \I_k$, where $I_i = [d_i, a_i]$, $D$ contains 1's at every position $d_i$, and $A$ contains 1's at every position $a_i$.
Figure~\ref{fig:2bitvectors-example} depicts this representation.

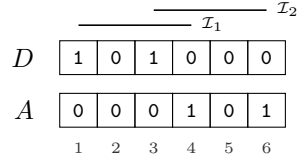
\begin{figure}
	\centering
	\begin{tikzpicture}
		\drawintervalsubfig{{1/4/1,3/6/2}}{0}{1}
	\end{tikzpicture}
	\caption{Representation of a set of non-nested interval using two \bvs, one for departures and the other for arrival timestamps. In this example, a set containing the intervals $[1, 4]$ and $[3, 6]$ is represented by the first \bv containing \texttt{1}'s at position $1$ and $3$, and the second \bv containing \texttt{1}'s at positions $4$ and $6$. Note that both \bvs must have the same number of \texttt{1}'s, otherwise, there would be an interval with missing values for departure or arrival.}
	\label{fig:2bitvectors-example}
\end{figure}

\subsection{Query algorithms}

Algorithms~\ref{alg:2bitvectors-2} and~\ref{alg:2bitvectors-3} answers the primitives \pprev{(D, A)}{t} and \pnext{(D, A)}{t}, respectively.
In order to find a previous interval, at line~1, Algorithm~\ref{alg:2bitvectors-2} first counts in $j$ how many 1's exist up to position $t$ in $A$.
If $j = 0$, then there is no interval $I = [t^-, t^+]$ such that $t^+ \leq t$, therefore, it returns nil.
Otherwise, at lines~4 and~5, the algorithm computes the positions of the $j$-th 1's in $D$ and $A$ to compose the resulting intervals.
In order to find a next interval, at line~1, Algorithm~\ref{alg:2bitvectors-3} first counts in $j'$ how many 1's exist up to time $t - 1$ in $D$.
If $j' = \brank{1}{D}{len(D)}$, then there is no interval $I' = [t'^-, t'^+]$ such that $t' \leq t^-$, therefore, it returns nil.
Otherwise, at lines~4 and~5, the algorithm computes the positions of the $(j' + 1)$-th 1's in $D$ and $A$ to compose the resulting interval.

\begin{algorithm}
	\caption{\pprev{(D, A)}{t}}\label{alg:2bitvectors-2}
	\begin{algorithmic}[1]
		\State{$j \gets \brank{1}{A}{t}$}
		\If{$j = 0$}
		\State{\textbf{return} nil}
		\EndIf{}

		\State{$t^- \gets \bselect{1}{D}{j}$}
		\State{$t^+ \gets \bselect{1}{A}{j}$}
		\State{\textbf{return} $[t^-, t^+]$}
	\end{algorithmic}
\end{algorithm}

\begin{algorithm}
	\caption{\pnext{(D, A)}{t}}\label{alg:2bitvectors-3}
	\begin{algorithmic}[1]
		\State{$j \gets \brank{1}{D}{t - 1}$}
		\If{$j = \brank{1}{D}{len(D)}$}
		\State{\textbf{return} nil}
		\EndIf{}

		\State{$t^- \gets \bselect{1}{D}{j + 1}$}
		\State{$t^+ \gets \bselect{1}{A}{j + 1}$}
		\State{\textbf{return} $[t^-, t^+]$}
	\end{algorithmic}
\end{algorithm}

As \brank{1}{B}{i} and \bselect{1}{B}{j} on dynamic \bvs{} have time complexity $O(\log{\tau})$ using the data structure proposed by~\cite{prezza2017framework}, \pprev{(D, A)}{t} and \pnext{(D, A)}{t} have both time complexity $O(\log{\tau})$.

\subsubsection{Interval insertion}

Due to the property of non-containment of intervals, given a new interval $\I = [t_1, t_2]$, we must first assure that there is no other interval $\I'$ in the data structure such that $\I \subseteq \I'$, otherwise, $\I$ cannot be present in the set.
Then, we must find and remove all intervals $\I''$ in the data structure such that $\I'' \subseteq \I$.
Finally, we insert $\I$ by setting the $t_1$-th bit of \bv{} $D$ and the $t_2$-th bit of $A$.
Figure~\ref{fig:2bitvectors-insertion} illustrates the process of inserting new intervals.

\begin{figure}
	\centering\small
	\begin{tabular}{cccc}
		\begin{tikzpicture} \drawintervalsubfig{{2/6/1}}{1}{1} \end{tikzpicture}  & \begin{tikzpicture} \drawintervalsubfig{{2/6/1}}{0}{0} \end{tikzpicture}  \\
		inserting $\I_1 = [2, 6]$                                                                                & inserting $\I_2 = [1, 6]$                                                                                \\
		\begin{tikzpicture} \drawintervalsubfig{{1/5/3,2/6/1}}{1}{1} \end{tikzpicture} & \begin{tikzpicture} \drawintervalsubfig{{3/4/4}}{1}{0} \end{tikzpicture} \\
		inserting $\I_3 = [1, 5]$                                                                                & inserting $\I_4 = [3, 4]$                                                                                \\
	\end{tabular}
	\caption{Sequence of insertions using our data structure based on \bvs $D$ and $A$. In~(a), our data structure is empty, thus, the insertion of interval $\I_1 = [2, 6]$ results in setting the position $2$ in $D$ and $6$ in $A$. Then, in~(b), the new interval $\I_2 = [1, 6]$ encloses $\I_1$, therefore, the insertion is skipped. Next, in~(c), no interval encloses or is enclosed by the new interval $\I_3 = [1, 5]$, thus, it suffices to set the position $1$ in $D$ and $5$ in $A$. Finally, in~(d), the new interval $\I_4 = [3, 4]$ is enclosed by $\I_1$ and $\I_3$, thus both of them is removed by clearing the corresponding bits and then $\I_4$ is inserted by setting the position $3$ in $D$ and $4$ in $A$.}
	\label{fig:2bitvectors-insertion}
\end{figure}
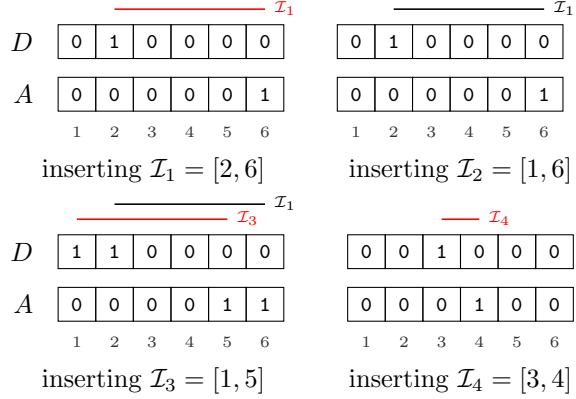

Algorithm~\ref{alg:2bitvectors-1} describes a simple process for the primitive \pinsert{(A, D)}{t_1}{t_2} in order to insert a new interval $\I = [t_1, t_2]$ into a set of non-nested intervals encoded as two \bvs{} $D$ and $A$.
At line~1, it computes how many 1's exist in $D$ prior to position $t_1$ by calling $r_d = \brank{1}{D}{t_1 - 1}$ and access the $t_i$-th bit in $D$ by calling $bit_d = \baccess{D}{t_1}$.
At line~2, it computes the same information with respect to the \bv{} $A$ and timestamp $t_2$ by calling $r_a = \brank{1}{A}{t_2 - 1}$ and $bit_a = \baccess{A}{t_2}$.
We note that the operations \brank{1}{B}{i} and \baccess{B}{i} can be processed in a single tree traversal using the dynamic \bv{} described in~\cite{prezza2017framework}.
If $r_d$ is less than $r_a + bit_a$, then there are more intervals closing up to timestamp $t_2$ than intervals opening before $t_1$, therefore, there is some interval $\I' = [d', a']$ such that $t_1 \leq d' \leq a' \leq t_2$, \textit{i.e.}, $\I \subseteq \I'$.
In this case, the algorithm stops, otherwise, it proceeds with the insertion.
When proceeding, if $r_d + bit_d$ is greater than $r_a$, then there are more intervals opening up to $t_1$ than intervals closing before $t_2$, therefore, there are $d = (r_d + bit_d) - r_a$ intervals $\I_i'' = [d_i'', a_i'']$, such that $d_i'' \leq t_1 \leq t_2 \leq a_i''$, \textit{i.e.}, $\I_i'' \subseteq \I$, that must be removed.
From lines~5 to~9, the algorithm removes the $d$ intervals that contain $I$ by iteratively unsetting their corresponding bits in $D$ and $A$.
In order to unset the $j$-th 1 in a \bv{} $B$, we first search for its position by calling $i = \bselect{1}{B}{j}$, then update $B[i] = 0$ by calling \bupdate{0}{B}{i}.
Thus, the algorithm calls \bupdate{0}{D}{\bselect{1}{D}{r_a + 1}} and \bupdate{0}{A}{\bselect{1}{A}{r_a + 1}} $d$ times to remove the $d$ intervals that closes after $r_a$.
Finally, at lines~10 and~11, the algorithm inserts $\I$ by calling \bupdate{1}{D}{t_1} and \bupdate{1}{A}{t_2}.
Note that both \bvs{} can grow with new insertions, thus we need to assure that both \bvs{} are large enough to accommodate the new $1$'s.
That is why the algorithm calls $ensureCapacity$ before setting the corresponding bits.
The $ensureCapacity$ implementation may call \binsert{0}{B}{len(B)} or \binsertw{0}{B}{len(B)} until $B$ has enough space.
Moreover, \brank{1}{B}{i} and \baccess{B}{i} operations can also receive positions that are larger than the actual length of $B$.
In such cases, these operations must instead return \brank{1}{B}{len(B)} and 0, respectively.

\begin{algorithm}
	\caption{\pinsert{(D, A)}{t_1}{t_2}}\label{alg:2bitvectors-1}
	\begin{algorithmic}[1]
		\State{$r_d \gets \brank{1}{D}{t_1 - 1}; \text{\ \,} bit_d \gets \baccess{D}{t_1}$}
		\State{$r_a \gets \brank{1}{A}{t_2 - 1}; \text{\ \,\,} bit_a \gets \baccess{A}{t_2}$}

		\If{$r_d \geq r_a + bit_a$}
		\If{$r_d + bit_d > r_a$}
		\State{$r_d^+ \gets r_d + bit_d$}
		\While{$r_d^+ > r_a$}
		\State{\bupdate{0}{D}{\bselect{1}{D}{r_a + 1}}}
		\State{\bupdate{0}{A}{\bselect{1}{A}{r_a + 1}}}
		\State{$r_d^+ \gets r_d^+ - 1$}
		\EndWhile{}
		\EndIf{}

		\State{$ensureCapacity(D, t_1); \text{\ \,} \bupdate{1}{D}{t_1}$}
		\State{$ensureCapacity(A, t_2); \text{\ \,} \bupdate{1}{A}{t_2}$}
		\EndIf{}
	\end{algorithmic}
\end{algorithm}

\begin{theorem}\label{theor:compact-insert-cost1}
	The update operation has worst-case time complexity $O(d\log\tau)$, where $d$ is the number of intervals removed.
\end{theorem}

\begin{proof}
    All operations on dynamic \bvs{} have time complexity $O(\log{\tau})$ using the data structure proposed by~\cite{prezza2017framework}.
	As the maximum length of each \bv{} is $\tau$, the cost of $ensureCapacity$ is amortized to $O(1)$ during a sequence of insertions.
	Therefore, the time complexity of \pinsert{(D, A)}{t_1}{t_2} is $O(d\log{\tau})$ since in the worst case Algorithm~\ref{alg:2bitvectors-1} removes $d$ intervals from line~6 to~9 before inserting the new one at lines~10 and~11.
\end{proof}

This simple strategy has a multiplicative factor on the number of removed intervals.
In general, as more intervals in $[1, \tau]$ are inserted, the number of intervals $d$ to be removed decreases, thus, in the long run, the runtime of this na\"{i}ve solution is acceptable.
However, when static \bvs{} are encoded sparsely as distances between consecutive 1's, it needs to decode/encode leaves $d$ times and thus runtime degrades severely.  
In the next section, we propose a new operation for dynamic \bvs{} using sparse static \bvs{} as leaves, \bunsetrange{B}, to replace this iterative approach and improve the time complexity of \pinsert{(D, A)}{t_1}{t_2} to $O(\log{\tau})$.

\subsection{New dynamic \bv operation to improve interval insertion}\label{ssec:new-bv-operations}

In this section, we propose a new operation \bunsetrange{B} for dynamic \bv{} using sparse static \bvs{} as leaves to improve the time complexity of \pinsert{(D, A)}{t_1}{t_2}.
This new operation clears all bits starting from the $j_1$-th 1 up to the $j_2$-th 1 in time $O(\log{\tau})$.
Our algorithm for \bunsetrange{B}, based on the split/join strategy commonly used in parallel programs~\cite{blelloch2016justJoin}, uses two internal functions \bsplit{N}{j} and \bjoin{N_1}{N_2}.
The \bsplit{N}{j} function takes a root node $N$ representing a dynamic \bv{} $B$ and splits its bits into two nodes $N_1$ and $N_2$ representing \bvs{} $B_1$ and $B_2$ containing, respectively, the bits in range $[1, \bselect{1}{B}{j} - 1]$ and $[\bselect{1}{B}{j}, len(B)]$.
The \bjoin{N_1}{N_2} function takes two root nodes $N_1$ and $N_2$, representing two \bvs{} $B_1$ and $B_2$ and constructs a new tree with root node $N$ representing a \bv{} $B$ containing all bits from $B_1$ followed by all bits from $B_2$.
The resulting trees for both functions must preserve the balancing properties of dynamic \bvs~\cite{prezza2017framework}.

Thus, given a dynamic \bv{} $B$ represented as a tree with root $N$, our algorithm for \bunsetrange{B} is described as follows.
First, the algorithm calls \bsplit{N}{j_1} in order to split the bits in $B$ into two nodes $N_{left}$ and $N_{tmp}$ representing two \bvs{} containing, respectively, the bits in range $[1, \bselect{1}{B}{j_1} - 1]$ and in range $[\bselect{1}{B}{j_1}, len(B)]$.
Then, it calls \bsplit{N_{tmp}}{j_2 - j_1} to split $N_{tmp}$ further into two nodes $N_{ones}$ and $N_{right}$ containing, respectively the bits in range $[\bselect{1}{B}{j_1}, \bselect{1}{B}{j_2} - 1]$, and $[\bselect{1}{B}{j_2}, len(B)]$.
The tree with root node $N_{ones}$ contains all 1's previously in the original dynamic \bv{} $B$ that should be cleared.
In the next step, the algorithm creates a new tree with root node $N_{zeros}$ containing $len(N_{ones})$ 0's to replace $N_{ones}$.
Finally, it calls \bjoin{\bjoin{N_{left}}{N_{zeros}}}{N_{right}} to join the trees with root nodes $N_{left}$, $N_{zeros}$, and $N_{right}$ into a final tree representing the original \bv{} $B$ with the corresponding 1's cleared.

Note that the tree with root $N_{ones}$ is still in memory, thus it needs some sort of cleaning.
The cost of immediately cleaning this tree would increase proportionally to the total number of nodes in $N_{ones}$ tree.
Instead, we keep $N_{ones}$ in memory and reuse its children lazily in other operations that request node allocations so that the cost of cleaning is amortized.
Moreover, even though we need to create a new \bv{} filled with zeros, this operation is performed in $O(1)$ time with a sparse implementation since only information about 1's is encoded.
We do not recommend using this strategy for a dense implementation, i.e, leaves represented as raw sequences of bits, since this last operation would run in time $O(\tau)$. 

Next we describe \bjoin{N_1}{N_2} and \bsplit{N}{j}.
The idea of \bjoin{N_1}{N_2} is to merge the root of the smallest tree with the correct node of the highest tree and rebalance the resulting tree recursively.
% During the process all steps must guarantee that the final tree is well-balanced in terms of the data structure described in~\cite{prezza2017framework}.

\begin{algorithm}[ht]
	\caption{\bjoin{N_1}{N_2}}\label{alg:join}
	\begin{algorithmic}[1]
		\If{$height(N_1) = height(N_2)$}
		\State{\Return $mergeOrGrow(N_1, N_2)$}
		\ElsIf{$height(N_1) > height(N_2)$}
		\State{$R \gets \bjoin{extractRightmostChild(N_1)}{N_2}$}
		\If{$height(R) = height(N_1)$}
		\State{\Return $mergeOrGrow(N_1, R)$}
		\EndIf{}
		\State{$insertRightmostChild(N_1, R)$}
		\State{\Return $N_1$}
		\Else{}
		\State{$R' \gets \bjoin{N_1}{extractLeftmostChild(N_2)}$}
		\If{$height(R') = height(N_2)$}
		\State{\Return $mergeOrGrow(R', N_2)$}
		\EndIf{}
		\State{$insertLeftmostChild(N_2, R')$}
		\State{\Return $N_2$}
		\EndIf{}
	\end{algorithmic}
\end{algorithm}

Algorithm~\ref{alg:join} details the \bjoin{N_1}{N_2} recursive function.
If $height(N_1) = height(N_2)$, at line~2, the algorithm tries to merge keys and pointers present in $N_1$ and $N_2$ if possible, or distributes their content evenly and grow the resulting tree by one level.
This process is done by calling $mergeOrGrow(N_1, N_2)$, which returns the root node of the resulting tree.
Instead, if $height(N_1) > height(N_2)$, at line~4, the algorithm first extracts the rightmost child from $N_1$, by calling $extractRightmostChild(N_1)$, and then recurses further passing the rightmost child instead.
The next recursive call might perform: a merge operation or grow the resulting subtree one level; therefore, the output node $R$ may have, respectively, height equals to $height(N_1) - 1$ or  $height(N_1)$.
If the resulting tree grew, \textit{i.e.}, $height(R) = height(N_1)$, then, at line~6, the algorithm returns the result of $mergeOrGrow(N_1, R)$.
Otherwise, if a merge operation was performed, \textit{i.e.}, $height(R) = height(N_1) - 1$, then, at line~7, it inserts $R$ into $N_1$ as its new rightmost child, and returns $N_1$.
Finally, if $height(N_1) < height(N_2)$, at line~10, the algorithm extracts the leftmost child from $N_2$ by calling $extractLeftmostChild(N_2)$ and recurses further passing the leftmost child instead.
Similarly, the root $R'$ resulted from the next recursive call might have height equals to $height(N_2) - 1$ or $height(N_2)$.
If $height(R') = height(N_2)$, then, at line~12, the algorithm returns the result of calling $mergeOrGrow(R, N_2)$, otherwise, if $height(R') = height(N_2) - 1$, then, at line~13, it inserts $R'$ into $N_2$ as its new leftmost child, and returns $N_2$.
Note that all subroutines must properly update keys describing the length and number of 1's of the \bv{} represented by the corresponding child subtree.
For instance, a call to $rightmost = extractRightmostChild(N)$ must decrement from the key associated with $N$ the length and number of 1's in the \bv{} represented by $rightmost$.

\begin{lemma}
	The operation \bjoin{N_1}{N_2} has time complexity $O(|height(N_1) - height(N_2)|)$.
\end{lemma}

\begin{proof}
	Algorithm~\ref{alg:join} descends at most $|height(N_1) - height(N_2)|$ levels starting from the root of the highest tree.
	At each level, in the worst case, it updates a node doing a constant amount of work equals to the branching factor of the tree.
	Therefore, the cost of \bjoin{N_1}{N_2} is $O(|height(N_1) - height(N_2)|)$.
\end{proof}

The idea of \bsplit{N}{j} is to traverse $N$ recursively while partitioning and joining its content properly until it reaches the node containing the $j$-th 1 at position \bselect{1}{B}{j}.
During the forward traversal, it partitions the current subtree in two nodes $N_1$ and $N_2$, excluding the entry associated with the child to descend.
Then, during the backward traversal, it joins $N_1$ and $N_2$, respectively, with the left and right nodes resulting from the recursive call.

\begin{algorithm}
	\caption{\bsplit{N}{j}}\label{alg:split}
	\begin{algorithmic}[1]
		\If{$N$ is leaf}
		\State{$(N_1, N_2) \gets partitionLeaf(N, j)$}
		\State{\Return{$(N_1, N_2)$}}
		\EndIf{}

		\State{$(N_1, child, N_2) \gets partitionNode(N, j)$}
		\State{$(N'_1, N'_2) \gets \Call{split}{child, j - ones(N_1)}$}
		\State{\Return{$(\Call{join}{N_1, N'_1}, \Call{join}{N'_2, N_2})$}}
	\end{algorithmic}
\end{algorithm}

The details of this function is shown in Algorithm~\ref{alg:split}.
From lines~1 to~3, the algorithm checks whether the root is a leaf.
If it is the case, it partitions the current \bv{} $B_1 \cdot{} b \cdot{} B_2$, where $b$ is the $j$-th 1, and returns two nodes containing, respectively, $B_1$ and $b \cdot{} B_2$.
Otherwise, from lines~4 to~6, the algorithm first finds the $i$-th child that contains the $j$-th 1 using a linear search and partitions the current node into three other nodes:
$N_1$, containing the partition with all keys and children in range $[1, i - 1]$;
$child$, which is the child node associated with position $i$; and
$N_2$, containing the partition with all keys and children in range $[i + 1, \ldots]$.
Then, at line~5, it recursively calls \bsplit{child}{j - ones(N_1)} to retrieve the partial results
$N'_1$ containing bits from $child$ up to the $j$-th 1; and
$N'_2$ containing bits from $child$ starting at the $j$-th 1 and forward.
Note that the next recursive call expects an input $j$ that is local to the root node $child$.
Finally, at line~6 it joins $N_1$ with $N'_1$ and $N'_2$ with $N_2$, and returns the resulting trees.

\begin{lemma}
	The operation \bsplit{N}{j} has time complexity $O(\log{\tau})$.
\end{lemma}

\begin{proof}
	As \bjoin{N_1}{N_2} has cost $O(|height(N_1) - height(N_2)|)$ and the sum of height differences for every level cannot be higher than the resulting tree height containing $n < \tau$ nodes, the time complexity of \bsplit{N}{j} is $O(\log{\tau})$.
\end{proof}

Furthermore, since \bjoin{N_1}{N_2} outputs a balanced tree when concatenating two already balanced trees, both trees resulting from the \bsplit{N}{j} calls are also balanced.

\begin{lemma}\label{lemma:unsetrangecost}
	The operation \bunsetrange{B} has time complexity $O(\log{\tau})$ when $B$ encodes leaves sparsely.
\end{lemma}

\begin{proof}
    The \bunsetrange{B} operation calls \bsplitname{} and \bjoinname{} twice.
    It must also create a new tree containing $\bselect{1}{B}{j_2 - 1} - \bselect{1}{B}{j_1}$ 0's to replace the subtree containing $j_2 - j_1$ 1's.
    If leaves of $B$ are represented sparsely, then the creation of a new tree filled with 0's costs $O(1)$ since the resulting tree only has a root node, with its only key having the current length (\bselect{1}{B}{j_2 - 1} - \bselect{1}{B}{j_1}), and an empty leaf.
    Therefore, as the cost of \bsplit{N}{j}, $O(\log{\tau})$, dominates the cost of \bjoin{N_1}{N_2}, the time complexity of \bunsetrange{B} is $O(\log{\tau})$.
\end{proof}

\begin{theorem}
    The primitive \pinsert{(D, A)}{t^-}{t^+} has time complexity $O(\log{\tau})$ when $D$ and $A$ encode leaves sparsely.
\end{theorem}

\begin{proof}
    Following from Theorem~\ref{theor:compact-insert-cost1} and Lemma~\ref{lemma:unsetrangecost}, the loop in Algorithm~\ref{alg:2bitvectors-1} that iteratively unset $d$ \bv{} bits can be substituted by a call to \bunsetrange{B}.
    As the cost of Algorithm~\ref{alg:2bitvectors-1} is dominated by this loop, its time complexity reduces to $O(\log{\tau})$. 
\end{proof}

\section{Experiments}\label{sec:experiments}

In this section, we conduct experiments to analyze the wall-clock time performance and the space efficiency of data structures when adding new information from synthetic datasets.
In Section~\ref{ssec:isolate-experiments}, we compare our compact data structure that maintain a set of non-nested intervals directly with an in-memory \Btree{} implementation storing intervals as keys.
For our compact data structure, we used dynamic \bvs~\cite{prezza2017framework} with leaves storing bits explicitly as arrays of integer words with words being $64$ bits long.
Internal nodes have a maximum number of pointers to children $m = 32$ and leaf nodes have static \bvs{} with maximum length $l = 4096$.
For the \Btree{} implementation we used $m = 32$ for all nodes.
In Section~\ref{ssec:integrate-experiments}, we compare the overall Temporal Transitive Closure (TTC) data structure using our new compact data structure with the TTC using the \Btree{} implementation for each pair of vertices.
All code is available at \url{https://bitbucket.org/luizufu/zig-ttc/src/master/}.

\subsection{Comparison of data structures for sets of non-nested intervals}\label{ssec:isolate-experiments}

For this experiment, we created datasets containing all $O(\tau^2)$ possible intervals in $[1, \tau]$ for $\tau \in [2^3, 2^{14}]$.
Then, for each dataset, we executed $10$ times a program that shuffles all intervals at random, and inserts them into the tested data structure while gathering the wall-clock time and memory space usage after every insertion.

\begin{figure*}
	\begin{center}
		\begin{tabular}{cc}
			\includegraphics[width=0.45\textwidth]{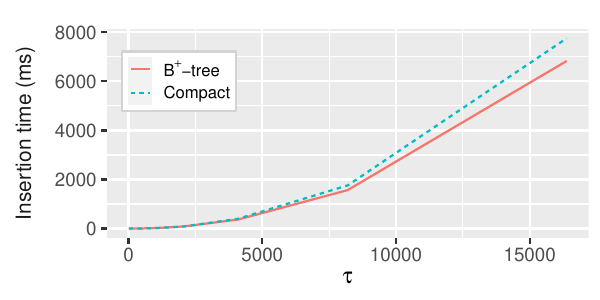} &
			\includegraphics[width=0.45\textwidth, page=60]{tree-comp-lifetime-random.pdf} \\
			(a) Overall                                                          &
			(b) Execution for $\tau = 2^{14}$                                              \\
		\end{tabular}
	\end{center}
	\caption{
		Comparison of incremental data structures to represent a set of non-nested intervals.
		In~(a), the overall average wall-clock time to insert all possible $O(\tau)$ intervals randomly shuffled into data structures.
		In~(b), the cumulative wall-clock time and the memory space usage to insert all possible $O(\tau)$ intervals randomly shuffled throughout a single execution.
		Note that the final wall-clock time of the execution described in~(b) was one of the $10$ executions with $\tau = 2^{14}$ used to construct~(a).
	}
	\label{fig:isolate-experiments}
\end{figure*}

Figure~\ref{fig:isolate-experiments}(a) shows the average wall-clock time to insert all intervals into the both data structures as $\tau$ increases.
Figure~\ref{fig:isolate-experiments}(b) shows the cumulative wall-clock time to insert all intervals and the memory usage throughout the lifetime of a single program execution with $\tau = 2^{14}$.
As shown in Figure~\ref{fig:isolate-experiments}(a), our new data structure slightly underperforms when compared with the \Btree{} implementation.
However, as shown in Figure~\ref{fig:isolate-experiments}(b), the wall-clock time have a higher overhead at the beginning of the execution (first quartile) and, after that, the difference between both data structures remains almost constant.
This overhead might be due to insertions of 0's at the end of the \bvs{} in order to make enough space to accommodate the rightmost interval inserted so far.
We can also see in Figure~\ref{fig:isolate-experiments}(b) that the space usage of our new data structure is much smaller than the \Btree{} implementation.
% Another benefit is that because bits are eagerly preallocated to insert the current right most interval, after some insertions the space usage rarely increases.
It is worth noting that, if the set of intervals is very sparse, maybe the use of sparse \bv{} as leaves could decrease the space since it does not need to preallocate most of the tree nodes, however, the wall-clock time could increase since at every operation leaves need to be decoded/unpacked and encoded/packed.

\subsection{Comparison of data structures for Time Transitive Closures}\label{ssec:integrate-experiments}

For this experiment, we created datasets containing all $O(n^2\tau)$ possible contacts fixing the number of vertices $n = 32$ and the latency to traverse an edge $\delta = 1$ while varying $\tau = [2^3, 2^{14}]$.
Then, for each dataset, we executed $10$ times a program that shuffles all contacts at random, and inserts them into the tested TTC data structure while gathering the wall-clock time and memory space usage after every insertion.

\begin{figure*}
	\begin{center}
		\begin{tabular}{cc}
			\includegraphics[width=0.45\textwidth, page=3]{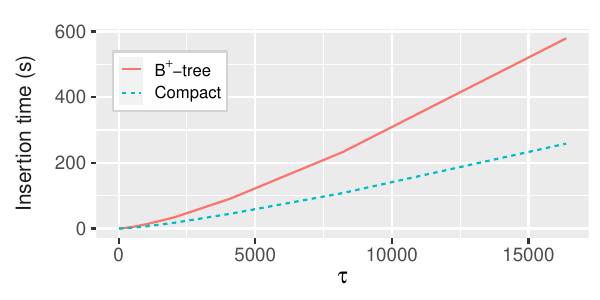} &
			\includegraphics[width=0.45\textwidth, page=190]{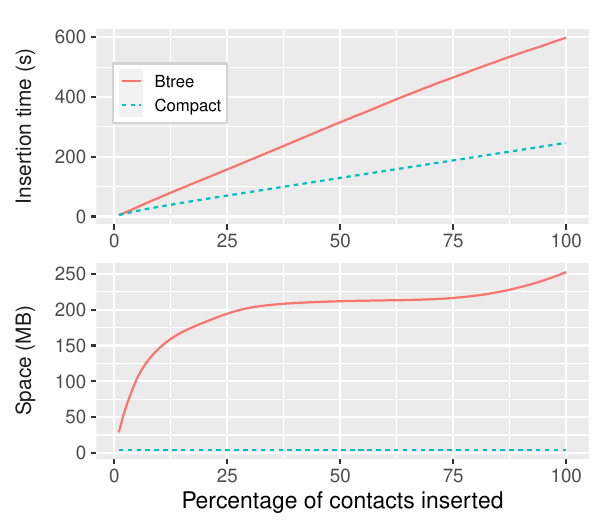} \\
			(a) Overall                                                                 &
			(b) Execution for $n = 32$ and $\tau = 2^{14}$                                 \\
		\end{tabular}
	\end{center}
	\caption{
        Comparison of Temporal Transitive Closures (TTCs) using incremental data structures to represent sets of non-nested intervals for each pair of vertices.
		In~(a), the overall average wall-clock time to insert all possible $O(n^2\tau)$ contacts randomly shuffled into data structures.
		In~(b), the cumulative wall-clock time and the memory space usage to insert all possible $O(n^2\tau)$ contacts randomly shuffled throughout a single execution.
		Note that the final wall-clock time of the execution described in~(b) was one of the $10$ executions with $\tau = 2^{14}$ used to construct~(a).
	}
	\label{fig:ttc-experiments}
\end{figure*}

Figure~\ref{fig:ttc-experiments}(a) shows the average wall-clock time to insert all contacts into the TTCs using both data structures as $\tau$ increases.
Figure~\ref{fig:ttc-experiments}(b) shows the cumulative wall-clock time to insert all contacts and the memory usage throughout the lifetime of a single program execution with $n = 32$ and $\tau = 2^{14}$.
As shown in Figure~\ref{fig:ttc-experiments}(a), the TTC version that uses our compact data structure in fact outperforms when compared with the TTC that uses the \Btree{} implementation for large values of $\tau$.
In Figure~\ref{fig:ttc-experiments}(b), we can see that the time to insert a contact into the TTC using our new data structure is lower during almost all lifetime, and the space usage followed the previous experiment comparing data structures in isolation.

\section{Conclusion and open questions}\label{sec:conclusions}

We presented in this paper an incremental compact data structure to represent a set of non-nested time intervals.
This new data structure is composed by two dynamic \bvs{} and works well using common operations on dynamic \bvs.
Among the operations of our new data structures are:
\pprev{(A, D)}{t}, which retrieves the previous interval $[t_1, t_2]$ such that $t_2 \leq t$ in time $O(\log\tau)$;
\pnext{(A, D)}{t}, which retrieves the next interval $[t_1, t_2]$ such that $t_1 \geq t$ also in time $O(\log\tau)$; and
\pinsert{(A, D)}{t_1}{t_2}, which inserts a new interval $\I = [t_1, t_2]$ if no other interval $\I'$ such that $\I \subseteq \I'$ exists while removing all intervals $\I''$ such that $\I \subseteq \I'' $ in time $O(d\log\tau)$, where $d$ is the number of intervals removed.
Moreover, we introduced a new operation \bunsetrange{B} for dynamic \bvs{} that encode leaves sparsely, which we used to improve the time complexity of our insert algorithm to $O(\log{\tau})$.

Additionally, we used our new data structure to incrementally maintain Temporal Transitive Closures (TTCs) using much less space
We used the same strategy as described in~\cite{brito2022dynamic}, however, instead of using Binary Search Trees (BSTs), we used our new compact data structure.
The time complexities of our algorithms for the new data structure are the same as those for BSTs.
However, as we showed in our experiments, using our new data structure greatly reduced the space usage for TTCs in several cases and, as they suggest, the wall-clock time to insert new contacts also improves as $\tau$ increases.

For future investigations, we conjecture that our compact data structure can be simplified further so that the content of both its \bvs{} are merged into a single data structure.
Our current insertion algorithm duplicates most operations in order to update both \bvs.
Furthermore, each of these operations traverse a tree-like data structure from top to bottom.
With a single tree-like data structure, our insertion algorithm could halve the number of traversals and, maybe, benefit from a better spatial locality.
In another direction, our algorithm for \pinsert{(A, D)}{t_1}{t_2} only has time complexity $O(\log{\tau})$ when both $A$ and $D$ encode leaves sparsely.
Perhaps, a dynamic \bv{} data structure that holds a mix of leaves represented densely or sparsely can be employed to retain the $O(\log{\tau})$ complexity while improving the overall runtime for other operations.
Lastly, we expect soon to evaluate our new compact data structure on larger datasets and under other scenarios; for instance, in very sparse and real temporal graphs.

\begin{paragraph}{Acknowledgements}
  This study was financed in part by Funda\c{c}\~{a}o de Amparo \`{a} Pesquisa do Estado de Minas Gerais (FAPEMIG) and the Coordena\c{c}\~{a}o de Aperfei\c{c}oamento de Pessoal de N\'{i}vel Superior - Brasil (CAPES) - Finance Code 001* - under the ``CAPES PrInt program'' awarded to the Computer Science Post-graduate Program of the Federal University of Uberl\^{a}ndia.
\end{paragraph}

\printbibliography{}

\end{document}